\documentclass[conference]{IEEEtran}
\usepackage{psfrag}
\usepackage{graphicx}
\usepackage{fancyhdr}

\usepackage{amsmath,amsthm}
\usepackage{amssymb}
\usepackage[all]{xy}
\usepackage{color}

\usepackage{graphics}
\usepackage{subfigure}
\renewcommand{\pi}{\mu}

\usepackage{cite}
\newcommand{\K}{\{1,\ldots, K\}}

 \newcommand{\SH}{\textnormal{Soft}}
 \newcommand{\Full}{\textnormal{Full}}
 
\newcommand{\M}{{M}}

\allowdisplaybreaks[4]

\def\1{{\mathbf 1}}        

\newtheorem{corollary}{Corollary}
\newtheorem{theorem}{Theorem}
\newtheorem{proposition}{Proposition}

\theoremstyle{remark}

\allowdisplaybreaks

\newcommand{\rt}[1]{{\color{black}{#1}}}
\newcommand{\mw}[1]{{\color{black}{#1}}}

\IEEEoverridecommandlockouts
\allowdisplaybreaks[3]

\begin{document}

\title{Complete Interference Mitigation Through Receiver-Caching in Wyner's Networks}

\author{ Mich\`{e}le Wigger$^{\dag}$, Roy Timo$^\ddag$, and Shlomo Shamai$^\dag\ddag$, \\
\small
$^{\dag}$LTCI CNRS, Telecom ParisTech, Universit\'e Paris-Saclay, 75013 Paris, France,
 michele.wigger@telecom-paristech.fr\\
 $^\ddag$Technische Universit\"{a}t M\"{u}nchen, roy.timo@tum.de, \qquad $^\dag \ddag $ Technion--Israel Institute of Technology, sshlomo@ee.technion.ac.il 

\thanks{
This work was supported by the Alexander von Humbold Foundation.}
}
\maketitle
\begin{abstract}
We present upper and lower bounds on the per-user multiplexing gain (MG) of Wyner's circular soft-handoff model and Wyner's circular  full model with cognitive transmitters and receivers with cache memories. The bounds are tight for cache memories with prelog \mw{$\mu\geq \frac{2}{3} D$} in the soft-handoff model and \mw{for $\mu\geq D$} in the full model, where $D$ denotes the number of possibly demanded files. In these cases the per-user MG of the two models is $1+ \frac{\mu}{D}$, the same as for non-interfering point-to-point links with caches at the receivers. Large receiver cache-memories thus allow to completely mitigate interference in these networks.
\end{abstract}

\section{Introduction}
We consider \rt{a} downlink \rt{communications problem} in cellular networks, where \rt{each} basestation (BS) BS\rt{is} connected to a \emph{cloud radio-access network (C-RAN)}. Such C-RANs \rt{can} facilitate cooperative interference cancellation at the \rt{BSs}, \rt{leading to} higher data rates and \rt{improved} energy efficiency~\cite{daiyu-2016}. With classical C-RAN architectures, \rt{however}, interference cancellation requires that either the backhaul links connecting the BSs with the C-RAN \rt{have large} capacity \cite{lapidothlevyshamaiwigger-2014-1} or that the C-RAN can store \rt{every} codebook and apply \rt{sophisticated} signal processing techniques~\cite{wiggertimoshamai16-1, ntranosmaddahalicaire-2015}. Both approaches are not feasible in practical systems. 

\rt{Researchers have also} proposed to replace the C-RAN with a  more complex \emph{fog radio-access network (F-RAN)}  \cite{pengyanzhangwang-2015}, which can \rt{perform more sophisticated} signal processing \mw{and be equipped with  \emph{caches}, where information  can be \rt{pre-stored}}. Improvements in rates of such F-RAN architectures are presented, \rt{for example}, in 
\cite{parksimeoneshamai-2016}.

In this work, we consider the classical C-RAN architecture 
\mw{and assume that the receiving mobiles have caches.  Noisy communication networks with receiver caching have been considered in\cite{timowigger-2015-1,hassanzadeherkipllorcatulino-2015,ghorbelkobayashiyang-2015,zhangelia-2015}.  Previous works on C-RANs or F-RANs assigned caches to the BSs, e.g., \cite{ugurawansezgin-2015,azarisimeonespagnolinitulino-2016, maddahaliniesen-2015-1,pooyaabolfazlhossein-2015}. }
Notice that communicating cache contents from a central server to the mobiles is hardly more demanding than to the BSs when the communication happens during periods of weak network congestions, e.g., over night. This is possible when file popularities vary only slowly in time, which we will assume in this paper.

We model the communication from the base-stations to the mobiles by  \emph{Wyner's  circular soft-handoff network} and \emph{Wyner's circular full network}. 
We \rt{will} show that when the receiving mobiles have sufficiently large cache memories, then  in the high signal-to-noise ratio (SNR) regime a \rt{combination} of
\begin{itemize} 
\item coded caching~\cite{maddahaliniesen-2014-1},
\item  interference-cancelation at the receivers that is facilitated by the cache contents, and 
\item broadcast transmission at the transmitters,
\end{itemize} 
allows \rt{one} to completely mitigate \rt{interference}. We present schemes that achieve the same \emph{full} multiplexing gain (MG) \rt{as interference-free} point-to-point links with caches at the receivers. 
For our scheme it suffices that each BS downloads only the messages intended to \mw{closeby mobiles}. The links from the C-RAN to the BSs can thus be of low rate, and  there is no need for sophisticated signal processing or  storing  codebooks at the C-RAN. Equipping mobile terminals in cellular networks with cache memories thus allows to substantially decrease backhaul costs and signal processing in C-RANs in the regime of high data rates  when almost all interference needs to be mitigated. 

\mw{The work most related to ours is~\cite{naderializadehmaddah-aliavestimehr-2016}, which also considers the MG of an interference network with receiver caching. In contrast to the paper here, in \cite{naderializadehmaddah-aliavestimehr-2016} the BSs have to decide already in advance which messages they  wish to download. The amount of downloaded data thus increases with the library size and is much larger than in our case.}

\section{Problem Description}
\label{sec:DescriptionOfTheProblem}


\begin{figure*}[htb!]
\centering
\includegraphics[width=0.75\textwidth]{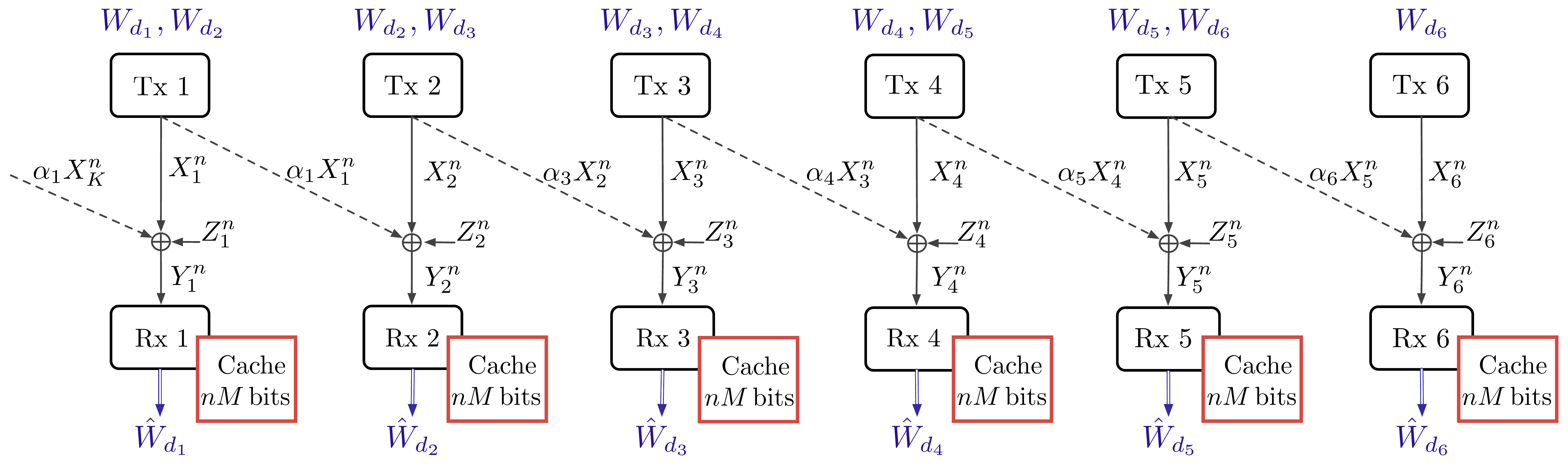}
\caption{Wyner's circular soft-handoff network for $K=6$. Each receiver has a cache of $n\M$ bits. During the delivery phase each transmitter knows the messages demanded by the two receivers that observe its signal.}
\label{fig:delivery}
\vspace{-3mm}
\end{figure*}
We consider two models for cellular communication with $K$ transmitter/receiver pairs:

\textit{Wyner's Circular  Soft-Handoff Model (fig.~\ref{fig:delivery}):}
The signal sent at transmitter (Tx)~$k$ is observed at receiver (Rx)~$k$ and Rx~$k+1$. Thus, if at time $t$, Tx~$k$ and Tx~$k-1$ send the real-valued symbols $X_{k,t}$ and $X_{k-1,t}$, then Rx~$k$ observes
\begin{equation}
\label{channel}
Y_{k,t}= X_{k,t} + \alpha_k X_{k-1,t} +Z_{k,t},
\end{equation}
where $\{Z_{k,t}\}$ is a sequence of  independent and identically distributed (i.i.d.) real standard Gaussians;  $\alpha_k\neq 0$ is a given real number; and 
$X_{0,t}= X_{K,t}.$

\textit{Wyner's Circular  Full Model:}
In this model, the signal sent at Tx~$k$ is also observed at Rx~$k-1$. Thus, at time~$t$, 
\begin{equation}
\label{channel_sym}
Y_{k,t}= X_{k,t} + \alpha X_{k-1,t} + \alpha X_{k+1,t} +Z_{k,t},
\end{equation}
where for simplicity we assume equal cross-gain $\alpha\neq 0$ for all receivers and $X_{K+1,t}=X_{1,t}$.

In both models, the inputs have to satisfy a symmetric average block-power constraint $P>0$: 
\begin{equation}\label{eq:power}
\frac{1}{n} \sum_{t=1}^n X_{k,t}^2 
\leq P, \quad \forall k \in \K,\quad \textnormal{almost surely}.
\end{equation}

\rt{We consider a} library of \rt{$D\geq 6$} messages:
\begin{equation}\label{eq:library}
\textnormal{Library:} \hspace{1cm}W_{1}, W_2, \ldots, W_{D},
\end{equation}
\rt{which are} independent and uniformly distributed over the set $\mathcal{W}\triangleq\{1,\ldots, \lfloor 2^{n R}\rfloor\}$.
Each receiver requests one of the files from the library. \rt{These requests, however, are not known prior to communications.}

Finally, each receiver has a  separate dedicated storage space, a \emph{cache}, of size $n{M}$ bits. A central server, who has access to the entire library \eqref{eq:library}, transmits dedicated caching information to each receiver, which this latter stores in its cache memory.  This transmission takes place during an idle network period before the actual communication starts. As a consequence the transmission is  error-free, but the transmitted information cannot depend on the receivers' actual  demands, which are still unknown at this moment.

We describe the placement of the cache information and the general communication in more detail. There are three phases: 

{\textit{Phase~1 --- Caching:}} A central server stores an arbitrary function of the entire library in this Rx~$k$'s cache, for $k\in\{1,\ldots, K\}$. \rt{Let} $\mathbb{V}_k$ \rt{denote} the content in Rx~$k$'s cache,
\begin{equation}
\mathbb{V}_k = \phi_k(W_1, \ldots, W_D), 
\end{equation}
for \rt{some} $\phi_k\colon \mathcal{W}^{D} \to \{1,\ldots, \lfloor 2^{n{M}}\rfloor\}.$

{\textit{Phase~2 --- Download from Server:}} Each Rx $k\in\{1,\ldots, K\}$ communicates its demand $d_k\in\{1,\ldots, D\}$ to its adjacent\footnote{Tx~$j$ is adjacent to Rx~$k$ if the transmit signal of Tx~$j$ is observed at Rx~$k$.} transmitters, which then download the desired message $W_{d_k}$ from the central server. 
So,  after the server download phase, in the soft-handoff model: 
\begin{equation}\label{eq:server_download}
\textnormal{Tx}~k \textnormal{ knows: } \quad (W_{d_k}, \ W_{d_{k+1}}),
\end{equation}
and in the full model:
\begin{equation}\label{eq:server_download}
\textnormal{Tx}~k \textnormal{ knows: }\quad  (W_{d_{k-1}}, \ W_{d_k}, \ W_{d_{k+1}}),
\end{equation}

{\textit{Phase~3 --- Delivery to Users:}}
The transmitters communicate the demanded messages 
$W_{d_1}, W_{d_2}, \ldots, W_{d_K}$
to Rx~$1, \ldots, K$, respectively. 
We assume that each transmitter and each receiver in the network knows the demand vector\footnote{Communicating $\mathbf{d}\in\{1,\ldots, D\}^K$ to all terminals requires zero rate.}
\begin{equation}
\mathbf{d} := (d_1, \ldots, d_K).
\end{equation}
(Tx~$k$ however only knows messages $W_{d_k}$ and $W_{d_{k+1}}$, \eqref{eq:server_download}.)

In the soft-handoff model, Tx~$k$ computes its inputs as 
\begin{equation}
X_k^n = f_{\SH,k}^{(n)}(W_{d_k}, W_{d_{k+1}}, \mathbf{d}),
\end{equation}
and in the full model as
\begin{equation}
X_k^n = f_{\Full,k}^{(n)}(W_{d_{k-1}}, W_{d_k}, W_{d_{k+1}}, \mathbf{d}).
\end{equation}
The encoding functions $f_{\SH,k}^{(n)} \colon \mathcal{W}^2\times \{1,\ldots, D\}^K\to \mathbb{R}^n$ and  $f_{\Full,k}^{(n)} \colon \mathcal{W}^3\times \{1,\ldots, D\}^K\to \mathbb{R}^n$ have to  satisfy~\eqref{eq:power}.

Each Rx~$k$ guesses its demanded message $W_{d_k}$
as
\begin{equation}
\hat{W}_{d_k} = g_k^{(n)}(Y_k^n, \mathbb{V}_k, \mathbf{d}),
\end{equation}
for some decoding function $g^{(n)}_k \colon \mathbb{R}^n \times \{1,\ldots, \lfloor 2^{n{M}} \rfloor \} \times \{1,\ldots, D\}^K \to \mathcal{W}.$

An error occurs in the communication whenever
\begin{equation}
\hat{W}_{d_k} \neq W_{d_k} \qquad \textnormal{for \rt{any} k}\in\{1,\ldots, K\}.
\end{equation}

Given power constraint $P$, a rate-memory pair $(R, \mathcal{M})$ is said \emph{achievable}, if there exist encoding and decoding functions with vanishing probabilities of error as $n\to \infty$.

We perform a high-SNR analysis of our setup where $P\gg 1$. 
A per-user  rate-memory MG pair $(\mathcal{S}, \mu)$ is said achievable, if for each power $P$ there exists an achievable rate-memory pair $(R_P, \M_P)$ so that 
\begin{subequations}
\begin{IEEEeqnarray}{rCl}
\varlimsup_{K\to \infty}\varlimsup_{P\to \infty} \frac{1}{K}  \cdot \frac{R_P }{\frac{1}{2} \log(1+P)}\geq \mathcal{S} \\
\varlimsup_{P\to \infty}  \cdot \frac{\M_P }{\frac{1}{2} \log(1+P)}\geq \mathcal{\mu}.
\end{IEEEeqnarray}
\end{subequations}


We are  interested in the \emph{per-user rate-memory MG tradeoff}
\begin{equation*} 
\mathcal{S}^*(\mu):= \sup\Big\{ \mathcal{S}\geq 0 \colon (\mathcal{S}, \ \mu) \textnormal{ an achievable MG pair}\Big\}.\end{equation*}
We use subscripts $_\SH$ and $_\Full$ to differentiate the per-user rate-memory MG tradeoffs of the two models.

\section{Results}
 
Let 
\begin{equation}
\mathcal{S}_{\SH,\textnormal{ach}}(\mu):= \begin{cases} \frac{2}{3} +\frac{3}{2} \frac{\mu}{D}&  \quad \textnormal{ if } 0 \leq \frac{\mu}{D} \leq \frac{2}{3} \\
1 + \frac{\mu}{D}&  \quad \textnormal{ if }  \frac{\mu}{D} \geq \frac{2}{3}.
\end{cases}
\end{equation}
\begin{theorem}[Soft-Handoff Model]\label{rate-memory-tradeoff}
For Wyner's circular soft-handoff model: 
\begin{equation}
\min\left\{ \frac{2}{3} + \frac{3\mu}{ D}\ , \ 1+ \frac{\mu}{D}\right\}\geq \mathcal{S}_\SH^*(\mu) \geq \mathcal{S}_{\SH,\textnormal{ach}}(\mu).
\end{equation}
\end{theorem}
\begin{IEEEproof}
Converse omitted. Direct part in Section~\ref{sec:direct1}.
\end{IEEEproof}

\begin{corollary}\label{cor:1}For the soft-handoff model and
$\frac{\mu}{D}\geq \frac{2}{3}$:
\begin{equation}
 \mathcal{S}_\SH^*(\mu) = 1+ \frac{\mu}{D}.
\end{equation}
\end{corollary} 

Let 
\begin{equation}
\mathcal{S}_{\Full,\textnormal{ach}}(\mu):= \begin{cases} \frac{2}{3} +\frac{4}{3} \frac{\mu}{D}&  \quad \textnormal{ if } 0 \leq \frac{\mu}{D} \leq 1 \\
1 + \frac{\mu}{D}&  \quad \textnormal{ if }  \frac{\mu}{D} \geq 1.
\end{cases}
\end{equation}
\begin{theorem}[Full Model]\label{rate-memory-tradeoff_full}
For Wyner's circular full model: 
\begin{equation}
\min\left\{ \frac{2}{3} + \frac{6\mu}{  D}\ , \ 1+ \frac{\mu}{D}\right\}\geq \mathcal{S}_\Full^*(\mu) \geq \mathcal{S}_{\Full,\textnormal{ach}}(\mu).
\end{equation}
\end{theorem}
\begin{IEEEproof}
Converse omitted. Direct part in Section~\ref{sec:direct2}.
\end{IEEEproof}

\begin{corollary}\label{cor:2}
For the full model and for $\frac{\mu}{D}\geq 1$:
\begin{equation}
 \mathcal{S}_\Full^*(\mu) = 1+ \frac{\mu}{D}.
\end{equation}
\end{corollary} 
\vspace{-4mm}
\begin{figure}[h!]
\subfigure[Soft-Handoff Model]{
\centering
\includegraphics[width=0.2\textwidth]{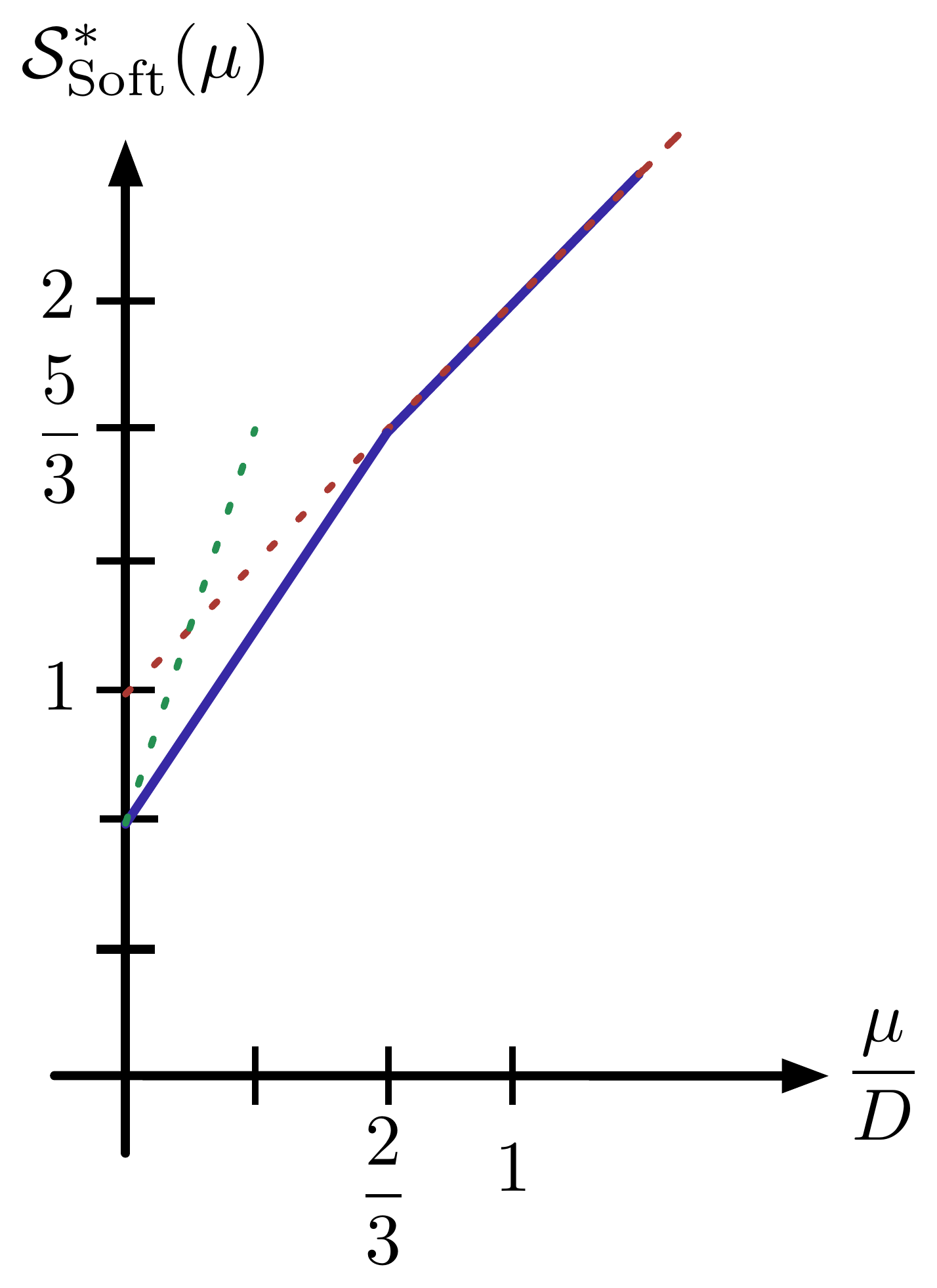}}
\hspace{.7cm}
\subfigure[Full Model]{\centering \includegraphics[width=0.22\textwidth]{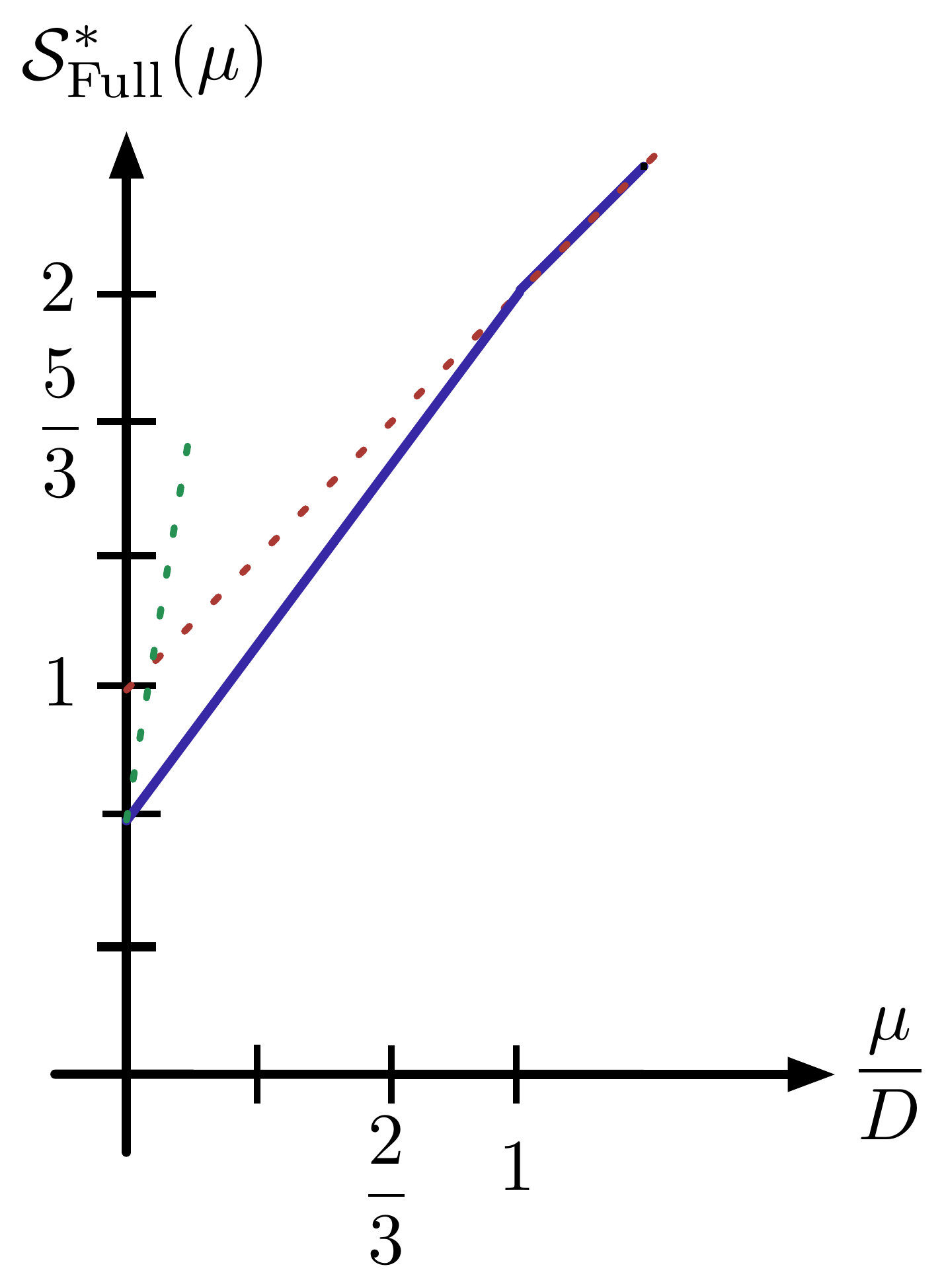}}
\caption{Upper and lower bounds on the per-user rate-memory MG tradeoffs.} 
\label{fig:tradeoff}
\end{figure}

Figure~\ref{fig:tradeoff}  plots the upper and lower bounds  on the per-user rate-memory MG tradeoff for the soft-handoff and the full circular Wyner model in Theorems~\ref{rate-memory-tradeoff} and \ref{rate-memory-tradeoff_full}. Without cache memories,  $\mu=0$, the per-user MG is $\frac{2}{3}$ for both models. 

In the regime $\frac{\mu}{D} \in [0, 2/3]$, the achievable per-user rate-memory MG tradeoffs $\mathcal{S}_{\SH,\textnormal{ach}}(\mu)$ and $\mathcal{S}_{\Full,\textnormal{ach}}(\mu)$ grow as a function of $\frac{\mu}{D}$  with slopes $\frac{3}{2}$ and $\frac{4}{3}$, respectively. The two slopes can be understood as follows. In the soft-handoff model, slope $\frac{3}{2}$ indicates that to achieve additional MG $\zeta\leq 1/3$ on the transmission over the network, \emph{two} cache memories of size $\zeta$ are required: one to cancel the interference caused by this additional transmission, and one to cancel  existing interference at the receiver that wants to decode the additional transmission. 
In the full model, 
\emph{three} cache memories of size~$\zeta$ are required to transmit an additional MG $\zeta\leq 1/3$ over the network: two cache memories are needed to cancel the symmetric interference caused by the additional transmission, and one to cancel  existing interference at the receiver that wants to decode the additional transmission. 

In both scenarios,  when  the cache memories are sufficiently large, the per-user rate-memory MG tradeoff equals $1+\frac{\mu}{D}$ (Cor.~\ref{rate-memory-tradeoff} and \ref{rate-memory-tradeoff_full}), the same as in an interference-free network. 

\section{Direct parts}\label{sec:Direct}
 \begin{proposition}\label{prop}
 Consider an arbitrary interference network. If rate-memory pair $(R, \M)$ is achievable, then for every positive $\Delta$ \rt{the} rate-memory pair
\begin{equation}
\Big(R+ \frac{\Delta}{D}, \; \M+  \Delta\Big)\quad  \textnormal{ is achievable}.
\end{equation} 
\end{proposition}
\begin{proof}For each $d\in\{1,\ldots, D\}$, split message $W_d$ into two submessages
\begin{equation*}
W_d= (W_d^{(1)}, W_d^{(2)}), 
\end{equation*} of rates $R$ and  $\frac{\Delta}{D}$. 
Apply the caching and delivery strategies of \rt{any scheme that achieves  the} rate-memory pair $(R, \M)$  to submessages $\big\{W_d^{(1)}\big\}_{d=1}^D$. Additionally, cache all submessages $W_{1}^{(2)}, \ldots, W_D^{(2)}$ at each and every receiver. This requires an additional memory size of $n\Delta$ bits.
\end{proof}

 \subsection{Direct part for Theorem~\ref{rate-memory-tradeoff}} \label{sec:direct1}

The scheme in \cite{lapidothlevyshamaiwigger-2014-1} achieves per-user rate-memory MG pair
\begin{equation}\label{eq:tradeoff1}
\mathcal S_\SH= \frac{2}{3} \quad \textnormal{and} \quad \mu=0.
\end{equation}

 In the following we show achievability of  
\begin{equation}\label{eq:tradeoff3}
\mathcal S_\SH= \frac{5}{3} \quad \textnormal{and} \quad \mu=\frac{2}{3}D.
\end{equation}

Time-sharing the schemes achieving  \eqref{eq:tradeoff1}  and \eqref{eq:tradeoff3} yields
\begin{equation}\label{eq:MG1}
 \mathcal{S}_\SH^*(\mu) \geq \frac{2}{3} + \frac{3}{2} \cdot \frac{\mu}{D}, \qquad \frac{\mu}{D}\in\Big[0,\ \frac{2}{3}\Big].
\end{equation}

Furthermore, by \eqref{eq:tradeoff3} and  proposition~\ref{prop}, 
\begin{equation}\label{eq:MG2}
 \mathcal{S}^*(\mu) \geq 1+\frac{\mu}{D}, \qquad \frac{\mu}{D}\in\Big[ \frac{2}{3}, \infty\Big),
\end{equation}
which concludes the proof.


We now explain the scheme achieving~\eqref{eq:tradeoff3}.
Define 
\begin{equation}\label{eq:alphamin}
\alpha_{\min}:= \min\{ 1, |\alpha_1|, \ldots, |\alpha_K|\}.
\end{equation}

Fix a small $\epsilon>0$ and let the  message rate $R$ be
\begin{equation}\label{eq:R2}
R= \frac{5}{3} \cdot \frac{1}{2} \log \left( 1+\alpha_{\min}^2(P-\epsilon)\right)-5 \epsilon.
\end{equation}
Split each message  into five independent submessages,
\begin{equation*}
W_d =\Big(W_d^{(1)}, W_{d}^{(2)}, W_d^{(3)}, W_{d}^{(4)}, W_d^{(5)}\Big), \quad d\in\{1,\ldots, D\},
\end{equation*}
where each of these submessages is of equal rate 
\[R/5=\frac{1}{3}\cdot \frac{1}{2}\log(1+\alpha_{\min}^2(P-\epsilon)) -\epsilon.\]
We form a sixth \emph{coded submessage}  
\begin{equation*}
W_{d}^{(6)} := W_d^{(1)} \bigoplus W_{d}^{(2)} \bigoplus W_d^{(3)} \bigoplus  W_{d}^{(4)} \bigoplus W_d^{(5)},
\end{equation*}
where $\bigoplus$ denotes the x-or operation performed on the bit-string representation of the messages. 
In our scheme, each Rx~$k$ will recover arbitrary five submessages pertaining to its desired message $W_{d_k}$. It will then x-or these five submessages to obtain an estimate of  the missing sixth submessage, which will allow it to form  an estimate of its desired $W_{d_k}$.

\textit{Caching:} The server makes the following cache assignment (see also figure~\ref{fig:scheme3_period1})
where $\mod$ denotes modulo operation: 
\begin{equation*}
\mathbb{V}_k= \begin{cases} \big(W_1^{(1)}, W_1^{(2)}, \ldots, W_{D}^{(1)}, W_{D}^{(2)}\big), & \textnormal{ if } k \mod 3 =1, \\
  \big(W_1^{(3)}, W_1^{(4)}, \ldots, W_{D}^{(3)}, W_{D}^{(4)}\big), & \textnormal{ if }k \mod 3 =2, \\
 \big(W_1^{(5)}, W_1^{(6)}, \ldots, W_{D}^{(5)}, W_{D}^{(6)}\big), & \textnormal{ if } k \mod 3 =0. \\\end{cases}
\end{equation*}
 
This cache assignment requires cache memory size
\begin{equation}\label{eq:M2}
\M = 2R/5=\frac{2}{3} \cdot \frac{1}{2}\log(1+\alpha_{\min}^2(P-\epsilon)) -2\epsilon.
\end{equation}

\textit{Delivery to Users:} Communication is split into periods 1--3 with $\lfloor \frac{n}{3} \rfloor$ consecutive channel uses each. 
In figures~\ref{fig:scheme3_period1}--\ref{fig:scheme3_period3} we depict for each period the submessages  sent by the transmitters and the submessages decoded at the receivers. 
\begin{figure}[h!]
\vspace{-2mm}
\centering
\includegraphics[width=1.0\columnwidth]{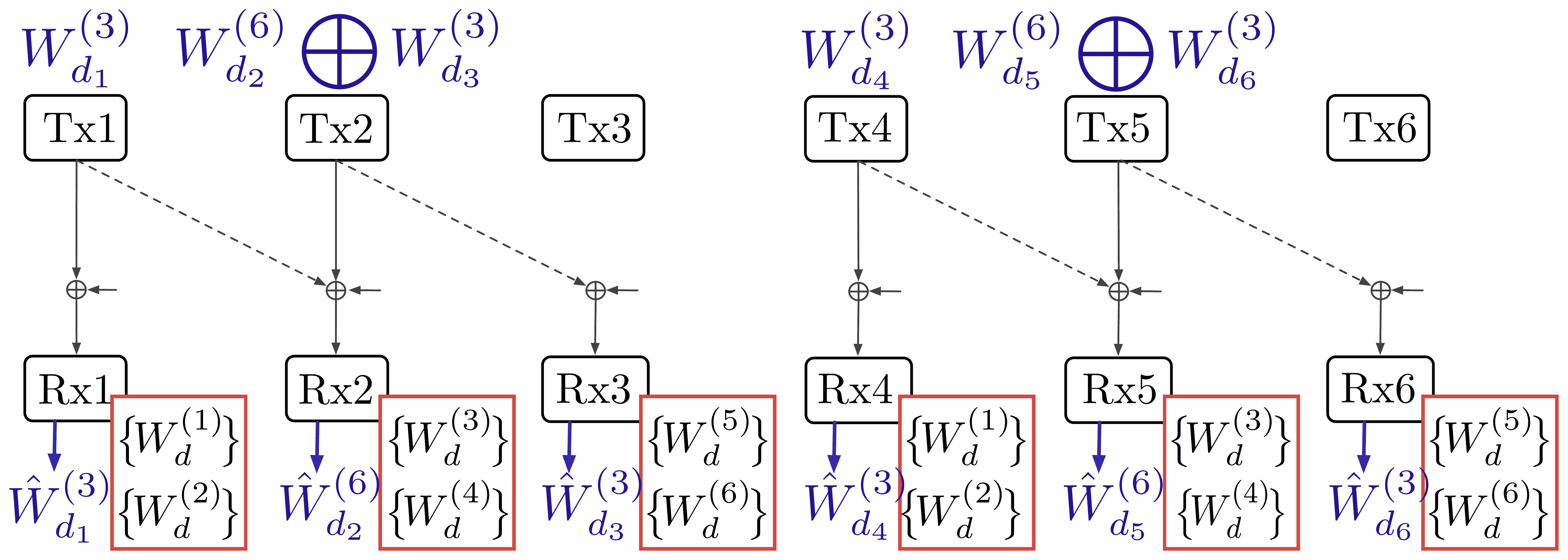}
\caption{Submessages sent during period~1.}
\label{fig:scheme3_period1}
\end{figure}
\vspace{-3mm}

\begin{figure}[h!]
\centering
\includegraphics[width=1\columnwidth]{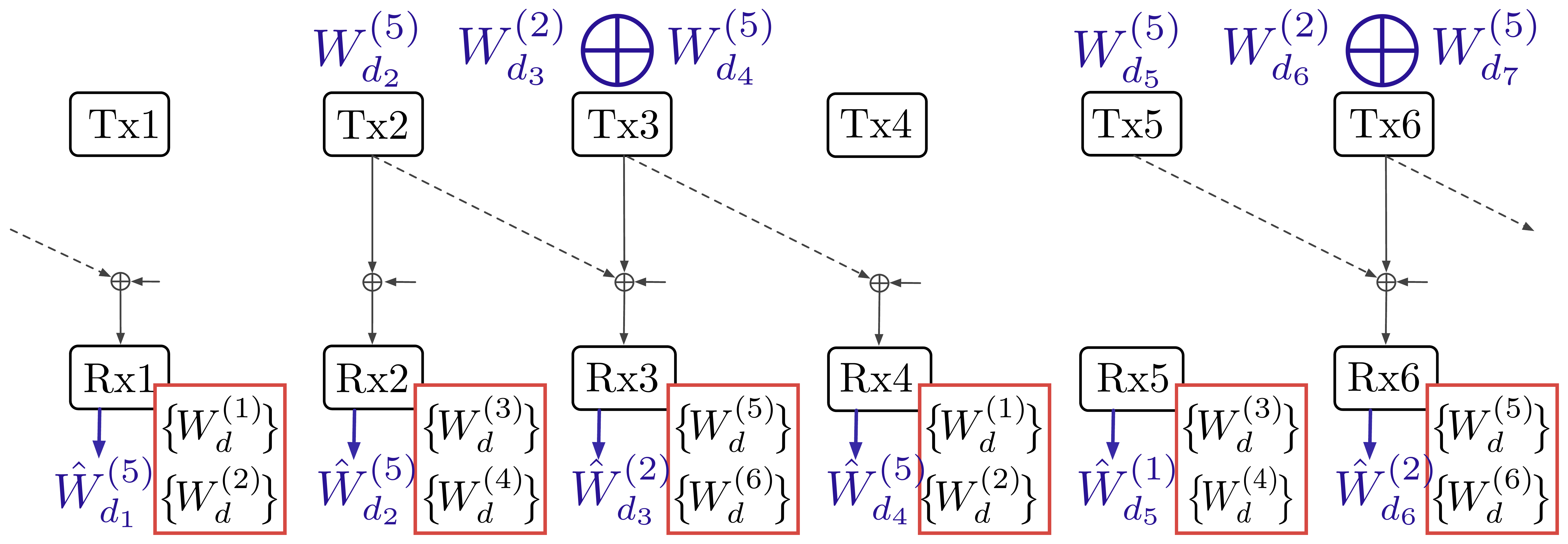}
\caption{Submessages sent during period~2.}
\label{fig:scheme3_period2}
\end{figure}
\vspace{-3mm}

\begin{figure}[h!]
\centering
\includegraphics[width=1\columnwidth]{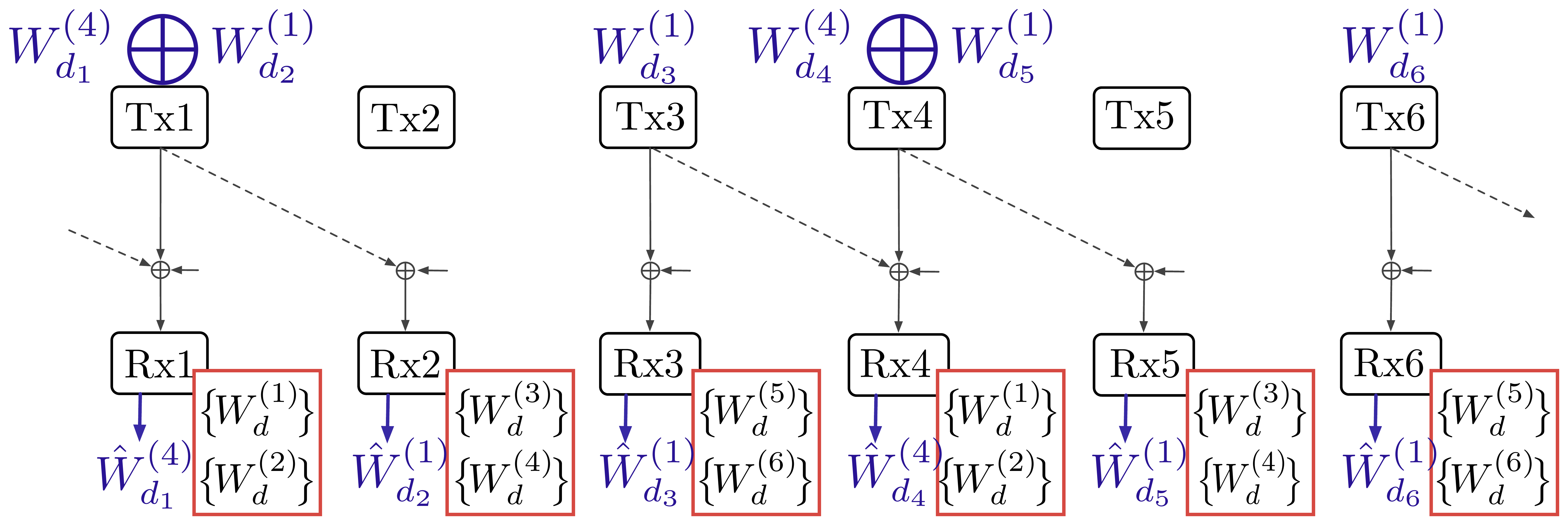}
\caption{Submessages sent during period~3.}
\label{fig:scheme3_period3}
\vspace{0mm}
\end{figure}

In each period, every third transmitter sets its channel inputs to $0^n$. 
In period~1, for example, 
\begin{equation*}
X_k^n= 0^n, \quad \textnormal{ if } k\mod 3 =0.
\end{equation*}
Similarly, in period~2, $X_k^n=0^n$ if $k\mod 3 =1$ and in period~3, $X_k^n=0^n$ if $k\mod 3=2$.
In all three periods, also Tx~$K$ sets its channel inputs to $0^n$.
This decomposes the network into non-interfering subnets of two active transmitters and three receivers.  (If $K \mod 3 =2$,  there is a last subnet with only Tx~$K-1$ and Rx~$K-1$.)

For example, in period~1, the first subnet
comprises Txs~1--2 and Rxs 1--3, see figure~\ref{fig:scheme3_period1}. 
We now describe communication over this first subnet during period~1, assuming that $d_1, d_2, d_3$ are all different. (Otherwise, a slightly changed scheme achieves the same performance.)

 Randomly draw for $k\in\{1,2\}$ a Gaussian codebook $\mathcal{C}_k:= \big\{X_k^{\lfloor n/3\rfloor}(1),\ldots, X_k^{\lfloor n/3\rfloor}(\lfloor 2^{nR/5}\rfloor )\big\}$, of length $\lfloor \frac{n}{3}\rfloor$,  rate-$3R/5$, and power $P-\epsilon$. Reveal the codebooks to all terminals of the subnet.

Tx~$1$ sends message $W_{d_1}^{(3)}$. To this end, it picks codeword 
\begin{equation}
X_{1}^{\lfloor n/3\rfloor}\big(W_{d_1}^{(3)}\big)
\end{equation}  from its codebook $\mathcal{C}_{1}$ and sends it over the network. 

Tx~$2$ forms  $W_{d_2}^{(6)} \bigoplus W_{d_3}^{(3)}$. It picks  codeword 
\begin{equation}
X_2^{\lfloor n/3\rfloor } \big(W_{d_2}^{(6)} \bigoplus W_{d_3}^{(3)} \big)
\end{equation}
from its codebook $\mathcal{C}_2$ and sends it over the network.

Rx~$1$ optimally  decodes submessage $W_{d_1}^{(3)}$ based on its outputs
\begin{equation}\label{eq:Y1}
Y_1^{\lfloor n/3\rfloor}= X_1^{\lfloor n/3\rfloor } \big(W_{d_1}^{(3)}\big) + Z_1^{\lfloor n/3\rfloor }.
\end{equation}

Rx~$2$ retrieves  submessage $W_{d_1}^{(3)}$ from its cache and forms
\begin{IEEEeqnarray}{rCl}\label{eq:Y2}
\tilde{Y}_2^{\lfloor n/3\rfloor}&:= &Y_2^{\lfloor n/3\rfloor }- \alpha_2 X_1^{\lfloor n/3\rfloor } \big(W_{d_1}^{(3)}\big) \nonumber \\
 & =& X_2^{\lfloor n/3\rfloor } \big(W_{d_2}^{(6)} \bigoplus W_{d_3}^{(3)}\big) + Z_1^{\lfloor n/3\rfloor }.
\end{IEEEeqnarray}
It optimally decodes the x-or submessage $W_{d_2}^{(1)} \bigoplus  W_{d_3}^{(3)}$ based on $\tilde{Y}_2^{\lfloor n/3\rfloor}$. Finally, it retrieves also $W_{d_3}^{(3)}$ from its cache and x-ors it with its guess of $W_{d_2}^{(1)} \bigoplus  W_{d_3}^{(3)}$. 

Recall that Tx~$3$ is deactivated. 
Rx~$3$ optimally decodes the x-or submessage $W_{d_2}^{(6)} \bigoplus  W_{d_3}^{(3)}$ based on its  observations
\begin{equation}\label{eq:Y3}
Y_3^{\lfloor n/3\rfloor }= \alpha_3 X_2^{\lfloor n/3\rfloor } \Big(W_{d_2}^{(6)}\bigoplus W_{d_3}^{(3)}\Big) + Z_3^{\lfloor n/3\rfloor }.
\end{equation}
It then retrieves $W_{d_2}^{(6)}$ from its cache and x-ors it with its guess of the x-or submessage $W_{d_2}^{(6)}\bigoplus W_{d_3}^{(3)}$. 

The other subnets and other periods are treated analogously. 
This way, after periods~1--3, each Tx~$k\in\{2,\ldots, K-1\}$ has decoded three submessages of its desired $W_{d_k}$, see figures~\ref{fig:scheme3_period1}--\ref{fig:scheme3_period3}. It uses these three decoded submessages and the two submessages of $W_{d_k}$  stored in its cache to  reconstruct the missing sixth submessage and the original message $W_{d_k}$.
 
 For example, Tx~2 has guessed $\hat{W}_{d_2}^{(6)}$ in period~1, $\hat{W}_{d_2}^{(5)}$ in period~2, and $\hat{W}_{d_2}^{(1)}$ in period 3 and  has stored $W_{d_2}^{(3)}$ and $W_{d_2}^{(4)}$ in its cache memory. It forms
 \begin{equation*}
 \hat{W}_{d_2}^{(2)}= \hat{W}_{d_2}^{(1)} \bigoplus W_{d_2}^{(3)} \bigoplus W_{d_2}^{(4)} \bigoplus \hat{W}_{d_2}^{(5)} \bigoplus \hat{W}_{d_2}^{(6)}, 
 \end{equation*}
 and produces as its final guess:
 \[\hat{W}_{d_2}=\Big( \hat{W}_{d_2}^{(1)}, \hat{W}_{d_2}^{(2)},   W_{d_2}^{(3)}, W_{d_2}^{(4)}, \hat{W}_{d_2}^{(5)}\Big).\]

\textit{Analysis:} Reconsider the first subnet of period 1.
Each receiver~1--3 optimally decodes either submessage $W_{d_1}^{(3)}$ or the x-or submessage  $\big(W_{d_2}^{(1)} \bigoplus  W_{d_3}^{(3)}\big)$ based on  interference-free Gaussian outputs \eqref{eq:Y1}--\eqref{eq:Y3}. Since the rate of these submessages,
$3R/5=\frac{1}{2} \log(1+\alpha_{\min}^2P) -3\epsilon$, is below the capacity of the channels \eqref{eq:Y1}--\eqref{eq:Y3}, the probability of decoding error tends to 0 as $n\to \infty$.  
If these decodings are error-free, then each Rx~1--3 produces the correct estimates in period~1. 

Analogous conclusions hold for all subnets formed in periods~1--3. This proves that in the proposed scheme the probability of error at Rx~$2$ to Rx~$K-1$ tends to 0 as $n\to \infty$. 

%

%


If $K$ is not a multiple of $3$, the same does not apply to   Rx~$1$ and Rx~$K$ because they have decoded less  than three submessages. 
 We can eliminate this edge effect by round-robin. Apply the same trick as before: split each message  into $K-2$ message-parts and combine these $K-2$ message-parts into  2 new x-or parts so that  from any $K-2$ parts one can reconstruct the original message.  Now, time-share $K$ instances of above scheme over $K$ equally long super-periods, where in super-period~$\ell\in\{1,\ldots, K\}$,  we transmit the $\ell$-th part of each message and we relabel Tx~$k$ and Rx~$k$ as Tx~$k-\ell$ and Rx~$k-\ell$, where indices need to be taken modulo $K$.  This way each receiver is a bad receiver, i.e.,  Rx~1 or Rx~$K$,  only a fraction $\frac{K-2}{K}$ of the time, when it simply ignores its outputs.

The round-robin modification makes that the probability of decoding error now tends to 0 as $n\to \infty$ at all $K$ receivers. But it requires that the rate of transmission is reduced by a factor $\frac{K-2}{K}$. This rate reduction however vanishes as $K\to \infty$. 


 \subsection{Direct part for Theorem~\ref{rate-memory-tradeoff_full}} \label{sec:direct2}

The scheme in \cite{lapidothlevyshamaiwigger-2014-1} achieves
\begin{equation}\label{eq:tradeoff1full}
\mathcal S_\Full= \frac{2}{3} \quad \textnormal{and} \quad \mu=0.
\end{equation}

We now describe a scheme that achieves \begin{equation}\label{eq:tradeoff3full}
\mathcal S_\Full= 2 \quad \textnormal{and} \quad \mu=D. 
\end{equation}
Fix a small $\epsilon>0$.  Let 
\begin{equation}\label{eq:R}
R=2 \cdot\Big(\frac{1}{2} \log(1+P-\epsilon)- \epsilon\Big),
\end{equation} 
and split each message $W_{d}$ into two parts $(W_{d}^{(1)}, W_d^{(2)})$ of equal rates $R/2$. Randomly draw for each $k\in\{1,\ldots, K\}$ a Gaussian codebook $\mathcal{C}_k:= \big\{X_k^n(1),\ldots, X_k^n(\lfloor 2^{nR/2}\rfloor )\big\}$ of length $n$, rate $R/2$ and power $P-\epsilon$.

\noindent \textit{Caching:}  Cache messages $W_{1}^{(1)}, \ldots, W_{D}^{(1)}$ at all odd receivers and messages $W_{1}^{(2)}, \ldots, W_{D}^{(2)}$ at all even receivers. 

\noindent \textit{Delivery to Users:} For $k\in\{1,\ldots,K\}$ odd, Tx~$k$ sends the codeword $X_k^n\big(W_{d_k}^{(2)}\big)$ from codebook $\mathcal{C}_k$ over the channel. 
 For $k\in\{1,\ldots, K\}$ even, Tx~$k$ sends the codeword $X_k^n\big(W_{d_k}^{(1)}\big)$  from codebook $\mathcal{C}_k$  over the channel.
 
 For $k$ odd, Rx~$k$ first takes messages $W_{d_{k-1}}^{(1)}$ and $W_{d_{k+1}}^{(1)}$ from its cache memory and forms 
 \begin{IEEEeqnarray}{rCl}\label{eq:c1}
 \hat{Y}_k^n&:= &Y_k^n - \alpha X_{k-1}^n\big(W_{d_{k-1}}^{(1)}\big)- \alpha X_{k+1}^n\big(W_{d_{k+1}}^{(1)}\big) \nonumber \\
 &= &X_k^n\big(W_{d_k}^{(2)}\big) +Z_k^n.
 \end{IEEEeqnarray}
 It optimally decodes $W_{d_k}^{(2)}$ from these interference-free outputs $\hat{Y}_k^n$.
 
 Similarly, for $k$ even, Rx~$k$ first takes messages $W_{k-1}^{(2)}$ and $W_{k+1}^{(2)}$ from its cache memory and 
 forms
 \begin{IEEEeqnarray}{rCl}\label{eq:c2}
 \hat{Y}_k^n&:= &Y_k^n - \alpha X_{k-1}^n\big(W_{d_{k-1}}^{(2)}\big)- \alpha X_{k+1}^n\big(W_{d_{k+1}}^{(2)}\big) \nonumber \\
 &= &X_k^n\big(W_{d_k}^{(1)}\big) +Z_k^n.
 \end{IEEEeqnarray}
 It optimally decodes $W_{d_k}^{(1)}$ from the new outputs $\hat{Y}_k^n$. 
 
 The probability of decoding error tends to 0 as $n\to \infty$ because the rate of communication $R/2$, \eqref{eq:R}, lies below the capacity of the created interference-free channels \eqref{eq:c1} and \eqref{eq:c2}. Since our scheme requires cache memory~$D\cdot\big(\frac{1}{2}\log(1+P-\epsilon)-\epsilon\big)$, this proves achievability of \eqref{eq:tradeoff3full}. 

Time-sharing the schemes achieving  \eqref{eq:tradeoff1full}  and \eqref{eq:tradeoff3full} yields
\begin{equation}\label{eq:MG1}
 \mathcal{S}_\Full^*(\mu) \geq \frac{2}{3} + \frac{4}{3} \cdot \frac{\mu}{D}, \qquad \frac{\mu}{D}\in\big[0,\ 1\big).
\end{equation}

Finally, by \eqref{eq:tradeoff3full} and proposition~\ref{prop}, 
\begin{equation}\label{eq:MG2}
 \mathcal{S}_\Full^*(\mu) \geq 1+\frac{\mu}{D}, \qquad \frac{\mu}{D}\in[ 1, \ \infty).\end{equation}

\end{document}